\documentclass[11pt]{article}

\usepackage{amsmath}
\usepackage{amsfonts}
\usepackage{graphicx}
\usepackage{enumerate}
\numberwithin{equation}{section}

\newtheorem{theorem}{Theorem}

\newtheorem{proof}{Proof}

\begin{document}

\title{Coding and Compression of Three Dimensional Meshes by Planes}

\maketitle

\begin{center}
{\large Rafik~ Aramyan *,  Gagik Mkrtchyan **,  Arman Karapetyan ***}
\end{center}

\begin{enumerate}
    \item[*]{\it
    Institute of mathematics
    of NAS of Armenia;\\
  {\small e-mail:  rafikaramyan@yahoo.com}}

  \item[**]{\it ixibit;\\
     {\small e-mail:
 mkrtchyan\_gagik@yahoo.com }}

 \item[***]{\it ixibit;\\
     {\small e-mail:
arman.karapetyan@ixibit.com}}

\end{enumerate}
\begin{abstract} The present paper suggests a new approach for geometric representation of 3D spatial models and provides a new compression algorithm for 3D meshes,  which is based on mathematical theory of convex geometry.
In our approach we represent a 3D convex polyhedron by means of planes, containing only its faces. This allows not to consider topological aspects of the problem (connectivity information among vertices and edges) since by means of the planes we construct the polyhedron uniquely. Due to the fact that the topological data is ignored this representation provides high degree of compression. Also planes based representation provides a compression of geometrical data because most of the faces of the polyhedron are not triangles but polygons with more than three vertices.
\end{abstract}

\section{Introduction}

In recent years more and more three dimensional (shortly 3D) spatial models become increasingly popular and available for advertising, World Wide Web, 3D laser scanning systems and etc. Highly detailed models are also commonly adopted in design of computer graphics. Mostly 3D graphical models are represented as complex polyhedral meshes, composed of topological and geometrical data. Topological data provides connectivity information among vertices (e.g., adjacency of vertices and edges), while geometrical attributes describe the position for each individual vertex. In terms of implementation, most of 3D graphical file formats consist of list of polygons, each of which is specified by its vertices indexes and their attributes.
	Generally speaking, real world 3D models are expensive to render, awkward to edit, and costly to transmit through networks since they contain tremendous number of vertices and polygons. For reduction of storage requirements and transmission bandwidth, it is desirable to compress these models with lossless and/or loss compression methods, which keep distortion within a tolerable level. This demands that meshes would be approximated with different resolutions and would be reduced by the coarse approximation through sequences of graphic simplifications.
Simplification and compression of 3D meshes data have been studied by many researchers. Most of the early works were focused on the simplification of graphical models.

\noindent Schroeder  \cite{S} proposed a decimation algorithm that significantly reduced the number of polygons required to represent an object.
Turk \cite{Tu}  presented an automatic method of creating surface models at several levels of detail from an original polyhedral description of a given object.
Hoppe \cite{Ho} address the mesh optimization problem of approximating a given point set by using smaller number of vertices under certain topological constraints.
Deering \cite{De} discuses the concept of generalized triangle mesh which compresses a triangle mesh structure.
Taubin  \cite{Ta} presented a topological surgery algorithm which utilized  two interleaving vertex and triangle trees to compress a model.
Hoppe  \cite{H} proposed a progressive mesh compression algorithm  that is applicable to arbitrary meshes.

\noindent The present paper suggests a new approach for geometric representation of 3D spatial models and provides a new compression algorithm for 3D meshes. In contrast to conventional representations here we suggest plane surface based representation for 3D meshes, which is based on mathematical theory of convex geometry.
In our approach we represent a 3D convex polyhedron by means of planes, containing only its faces. This allows not to consider topological aspects of the problem (connectivity information among vertices and edges) since by means of the planes we construct the polyhedron uniquely. Due to the fact that the topological data is ignored this representation provides high degree of compression. Also planes based representation provides a compression of geometrical data because most of the faces of the polyhedron are not triangles but polygons with more than three vertices.
For non convex 3D meshes we initially divide them into the groups of convex polyhedrons, and then each convex polyhedron is separately represented by its set of planes. Here we suggest an algorithms for division of non-convex 3D meshes, which divides into the convex parts by separation of convex and concave surface elements.

\noindent The features and advantages of our result will be more fully understood and appreciated upon consideration of its detailed description. First we need to describe a plane in 3D space.

\section{Representation of a plane}

We consider only oriented planes since 3D models can be represented by single sided surface. By definition, an oriented plane is a plane with specified normal direction.
An oriented plane will be denoted by $e$. Each oriented plane divides space into two hemispaces: $e_-$  and $e_+$. By $e_+$ we denote the hemispace on direction of the $e$ plane's normal, and by $e_-$ we denote the hemispace on its inverse direction (see Fig. 1). We consider the hemispace $e_-$ with its boundary, so it is closed.
An oriented plane $e$ can be represented by means of the following pair $(\omega, h)$.
Where $\omega$ is the spatial direction of its normal, which means that $\omega\in \mathbf
S^{2}$,  where $\mathbf S^2$ is the unit sphere in the three dimensional Euclidian space $\mathbb{R}^3$. And $h$ is the distance of $e$ from the origin O including its sign, which means that $h\in(-\infty,+\infty)$, and $h \geq 0$ if  the origin O belongs to $e_-$ and $h < 0$ if the origin O belongs to $e_+$  (see Fig. 2).

\begin{figure}[here]
\center{\includegraphics[width=9.8cm, height=7.5cm]{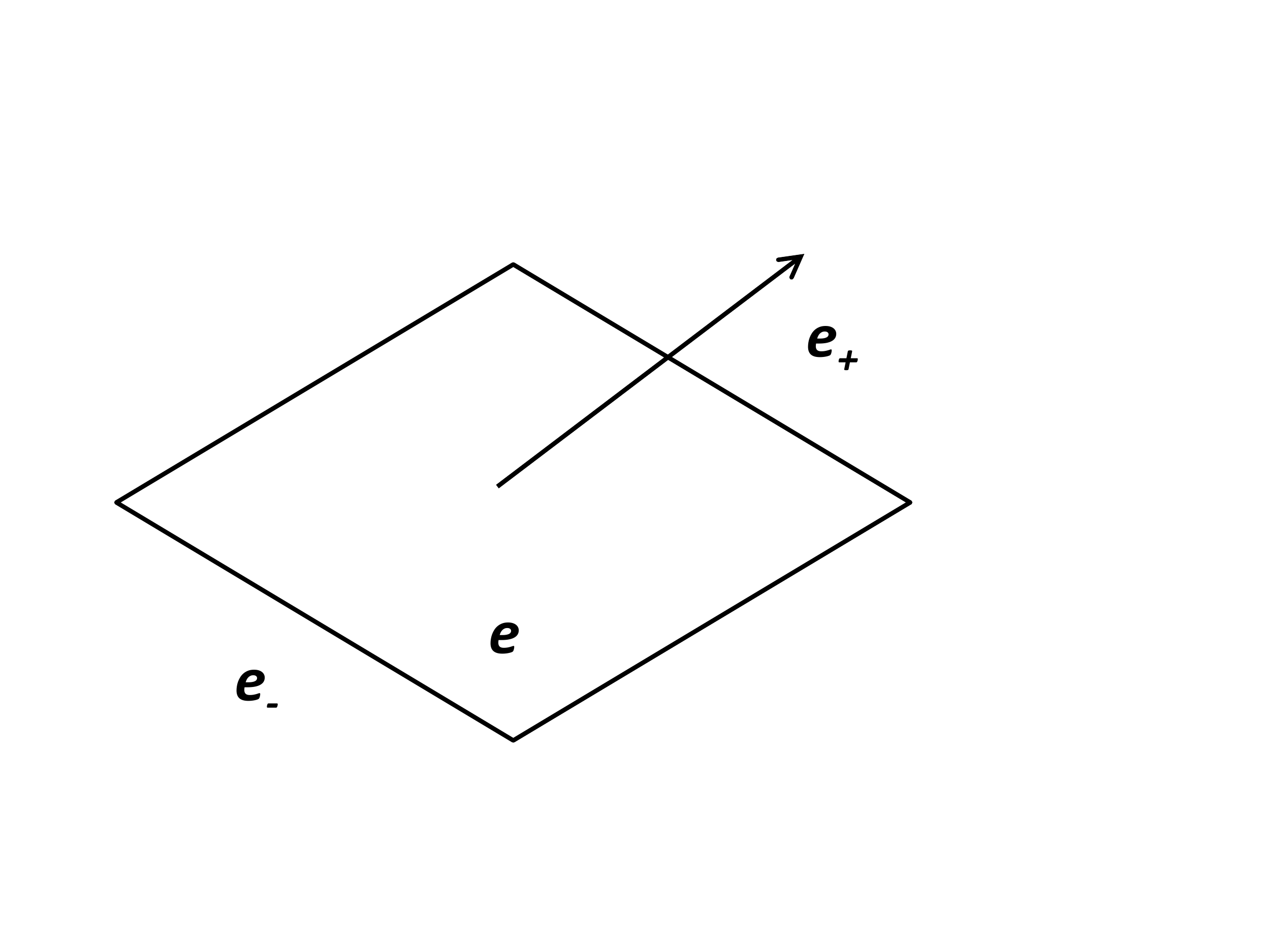}}
       \end{figure}
\centerline{\small{\bf Fig. 1.}}

\begin{figure}[here]
\center{\includegraphics[width=9.8cm, height=7.5cm]{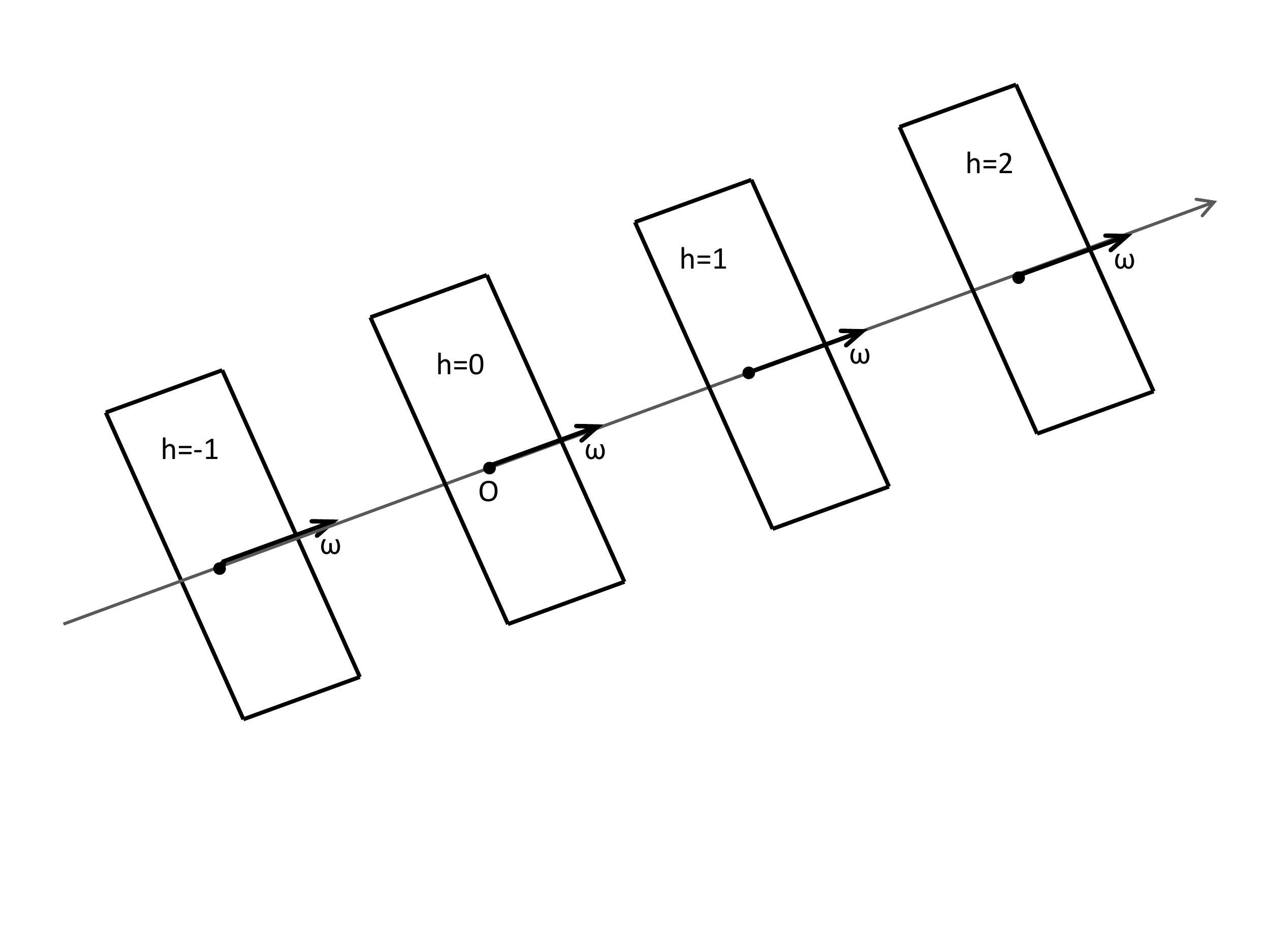}}
       \end{figure}
\centerline{\small{\bf Fig. 2.}}

\noindent Representation of a spatial direction.  Let $\omega\in \mathbf
S^{2}$ be a spatial direction, then by using spherical coordinates it can be represented as
$\omega=(\nu,\varphi)$, where $\nu$ is the angle between z-axis and $\omega$, and $\varphi$ is the angle between x-axis and the projection of the $\omega$ onto the $ xy$ plane. Note that the variable $\nu$ changes in the interval
$[0, \pi]$, while $\varphi$ changes in the interval $[0, 2\pi]$. Since spatial direction is represented by two numbers, and a plane is a pair of a spatial direction and distance $h$. Thus  any plane can be represented by three numbers.

\section{Representation of a convex polyhedron}
Let $\mathbf{P}$ be a convex polyhedron, its faces can be numbered by $i=1,2,...,n$. Let $e_i$ be an oriented plane containing $i$ -th face of the polyhedron  ($i = 1, 2,..., n$), which has outside directed normal and $\omega_i \in\mathbf
S^{2}$ be the normal of $e_i$  ($\mathbf
S^{2}$ is the unit sphere in 3D space) and $h_i$ be the distance (including sign) of the plane containing $i$ -th face of the polyhedron from the origin O, ($h_i\in (-\infty,+\infty)$).

\noindent  For a convex polyhedron $\mathbf{P}$ we will have a collection of oriented planes $\{e_i\}$,  $i=1,2,...,n$ or the collection of pairs $\{\omega_i,h_i\}$, $i=1,2,...,n$. This procedure allows representation of any convex polyhedron by a collection of planes (see Fig. 3).

\begin{figure}[here]
\center{\includegraphics[width=9.8cm, height=7.5cm]{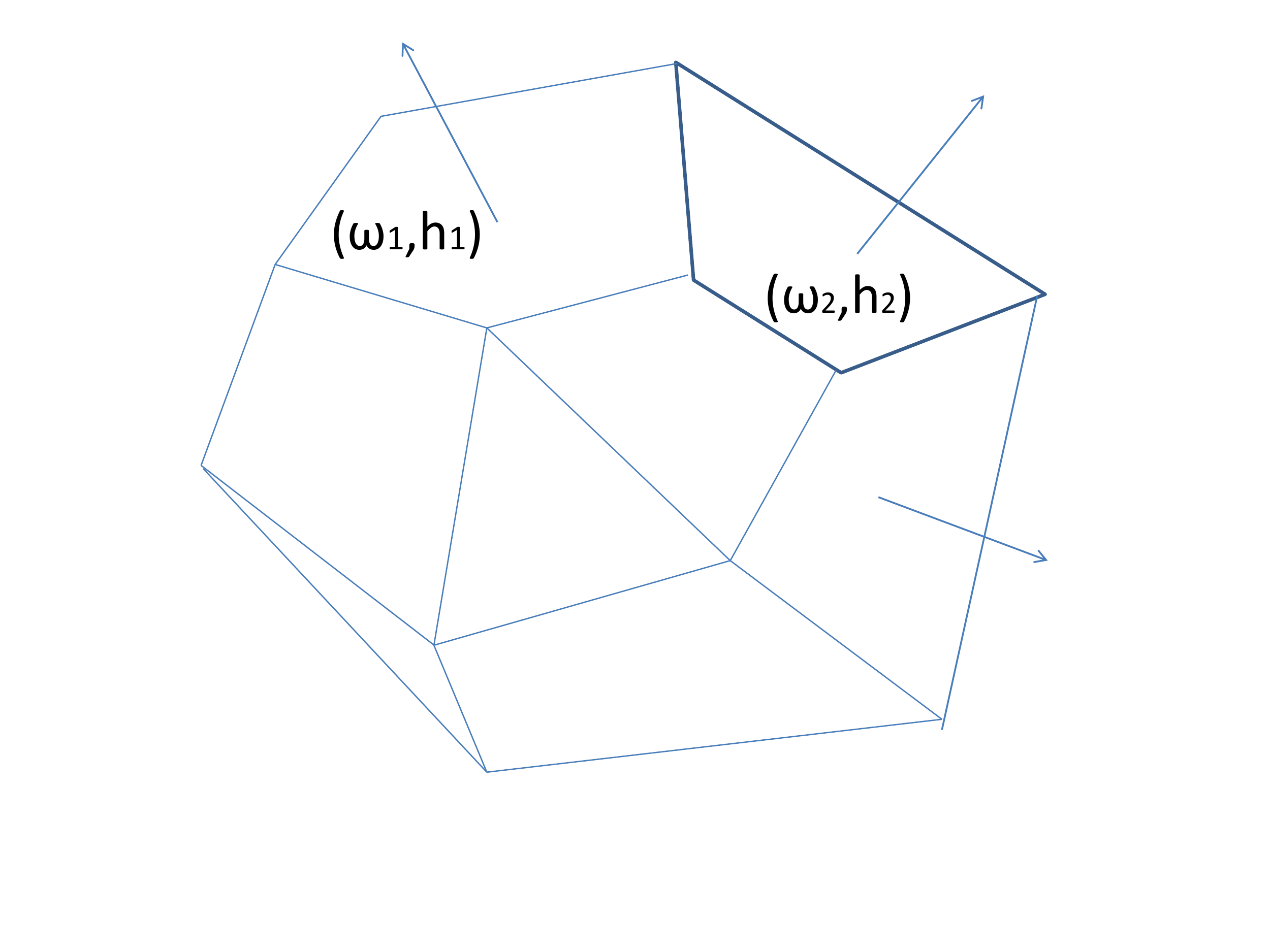}}
   \end{figure}
\centerline{\small{\bf Fig. 3.}}

\noindent As a result any convex polyhedron can be represented by a collection of triplets, since any plane can be represent by means of three numbers
\begin{equation}\label{1}
\mathbf{P}\longrightarrow \{e_i\} \Longleftrightarrow \{\omega_i,h_i\} \Longleftrightarrow  \{\nu_i,\varphi_i,h_i\}.
\end{equation}

\noindent The following theorem proves the uniqueness of this representation.

\begin{theorem} Let $\mathbf{P}$ be a convex polyhedron and $\{e_i\}$ be the collection of its oriented planes. We have
\begin{equation}\label{1.1}
\mathbf{P}= \cap_{i=1}^n \{e_i\}_- .
\end{equation}
Hence P is uniquely determined by its collection of oriented planes.\end{theorem}

\begin{proof} Let $\mathbf{P}$ be a convex polyhedron and $e_i$ be the  oriented plane containing $i$ -th face of the polyhedron ($i = 1, 2,..., n$), which has outside directed normal.
It follows from the convexity that the polyhedron P belongs to hemi space $\{e_i\}_-$  for each $i=1,2,…,n$. Hence the polyhedron P belongs to the intersection of that hemi spaces $\mathbf{P}\subset \cap_{i=1}^n \{e_i\}_-$. Now we have to conform, that $\mathbf{P}=\cap_{i=1}^n \{e_i\}_-$. Indeed, If a point $x$ does not belong to $\mathbf{P}$, then exists a number $i$ for which $x$ does not belong to hemi space $\{e_i\}_-$, hence it also does not belong to intersection $ \cap_{i=1}^n \{e_i\}_-$.
The proof is complete.\end{proof}

\noindent  It follows from the above theorem that any convex polyhedron can be uniquely represented by means of oriented planes containing its faces, i.e. by means of the system $ \{\omega_i,h_i\}$.
Note that in cases of rotation and/or translation of a polyhedron its new representation can be recalculated very easily.
Let a polyhedron $\mathbf{P}$, with corresponding system $ \{\omega_i,h_i\}$, is translated by a vector $\overrightarrow{a}$. We denote by $ \{\omega_i^a,h_i^a\}$ the new system of representation a $\overrightarrow{a}\mathbf{P}$.
\begin{theorem} Let $\mathbf{P}$ be a convex polyhedron and $ \{\omega_i,h_i\}$ be the collection of its oriented planes. Let $\overrightarrow{a}\mathbf{P}$ be the translation of $\mathbf{P}$ by a vector $\overrightarrow{a}$. We have
\begin{equation}\label{1.1}
 \{\omega_i^a,h_i^a\}=  \{\omega_i,h_i+\langle \omega_i, \overrightarrow{a} \rangle\},
\end{equation}
where $\langle \omega_i, \overrightarrow{a} \rangle$  is the scalar product of the vectors $\omega_i$   and  $\overrightarrow{a}$.
\end{theorem}

\begin{proof} It is easy to understand that after translation, the normal direction of $i$-th face does not change, and $h_i^a$ can be given by the following simple relations:
\begin{multline}\label{3}
\,\,\,\,\,\,\,\,\,\,\,\,\,\,\,\,\,\,\,\,\,\,\,\,\,\,\,\,\,\,\,\,\,\,\,\,\,\,\,\,\,\,\,\,\,\,\,\,\,\,\,\,\,\,\,\,\,\,\,\,\,\,\,\,\,\,\,\,\,\,\,\,\,\,\,\,\,\,\,\,\,\,\,\,\,\,\,\,\omega_i^t= \omega_i\\
h_i^t =   h_i+ \langle \omega_i, \overrightarrow{a} \rangle,\,\,\,\,\,\,\,\,\,\,\,\,\,\,\,\,\,\,\,\,\,\,\,\,\,\,\,\,\,\,\,\,\,\,\,\,\,\,\,\,\,\,\,\,\,\,\,\,\,\,\,\,\,\,\,\,\,\,\,\,\,\,\,\end{multline}
\end{proof}

\noindent Let a polyhedron $\mathbf{P}$, with corresponding system $ \{\omega_i,h_i\}$, is rotated with respect to the origin O. We denote by $ \{\omega_i^r,h_i^r\}$ the system of representation of $rot\mathbf{P}$ .

\begin{theorem} Let $\mathbf{P}$ be a convex polyhedron and $ \{\omega_i,h_i\}$ be the collection of its oriented planes. Let $rot\mathbf{P}$ be the rotation of $\mathbf{P}$ with respect to the origin O.. We have
\begin{equation}\label{1.2}
 \{\omega_i^r,h_i^r\}=  \{rot\omega_i,h_i\},
\end{equation}
where  $rot\,\omega_i$   is the  rotation of the direction $\omega_i$ by the same rotation.
\end{theorem}
\begin{proof}
It is easy to understand that after rotation, the distance of the plane containing $i$-th face does not changes, and  $rot\omega_i$ can be given by the following simple relations:
\begin{multline}\label{3}
\,\,\,\,\,\,\,\,\,\,\,\,\,\,\,\,\,\,\,\,\,\,\,\,\,\,\,\,\,\,\,\,\,\,\,\,\,\,\,\,\,\,\,\,\,\,\,\,\,\,\,\,\,\,\,\,\,\,\,\,\,\,\,\,\,\,\,\,\,\,\,\,\,\,\,\,\,\,\,\,\,\,\,\,\,\,\,\,\omega_i^r= rot\,\omega_i\\
h_i^r = h_i\,\,\,\,\,\,\,\,\,\,\,\,\,\,\,\,\,\,\,\,\,\,\,\,\,\,\,\,\,\,\,\,\,\,\,\,\,\,\,\,\,\,\,\,\,\,\,\,\,\,\,\,\,\,\,\,\,\,\,\,\,\,\,\,\,\,\,\,\,\,\,\,\,\,\,\,\,\,\,\,\,\,\,\,\,\,\,\,\,.\end{multline}
\end{proof}

\section{ An example plane based representation of 3D polyhedron}
Plane based representation for cube.
Let $\mathbf{P}$ be a cube, which six faces are numerated as shown in Fig. 4.
We denote by  $\omega_i$ the normal of $i$-th face, and the corresponding system of planes will be $\{ (\omega_1,1),  (\omega_2,1),  (\omega_3,0),  (\omega_4,0), (\omega_5,1),  (\omega_6,0)\}$. Now by using spherical coordinates the cube can be represented by the following system:
$$(90, 0,1), (90,90,1), (90,180,0), (90,270,0), (0,0,1), (180,0,0).$$
By using this representation the cube can be coded by means of  $6\cdot3 = 18$ number, which will require only $18\cdot4 = 72$ bytes for its storage.
For comparison with conventional coding, we represent the cube as an indexed triangular mesh. At first step we should numerate (create indexes) the vertices of the cube. We do it as shown in the Fig. 5 (note that each vertex is coded by three numbers, which are its x, y, z coordinates), and for storing it would be required $8\cdot 3\cdot 4 = 96$ bytes. Additionally should be stored the code for connectivity information among vertices as follows:
 $$(1,4,2), (2,4,3), (1,2,5), (2,6,5), (2,7,6), (2,3,7),$$
 $$ (3,8,7), (3,4,8), (1,4,8), (1,8,5), (5,6,8), (6,7,8),$$
which requires $12\cdot3\cdot·4 = 144 $ bytes. Thus by means of triangular mesh is needed

\noindent$96 + 144 = 240$  bites, which is three times more than the plane based storage requirement.

\begin{figure}[here]
\center{\includegraphics[width=9.8cm, height=7.5cm]{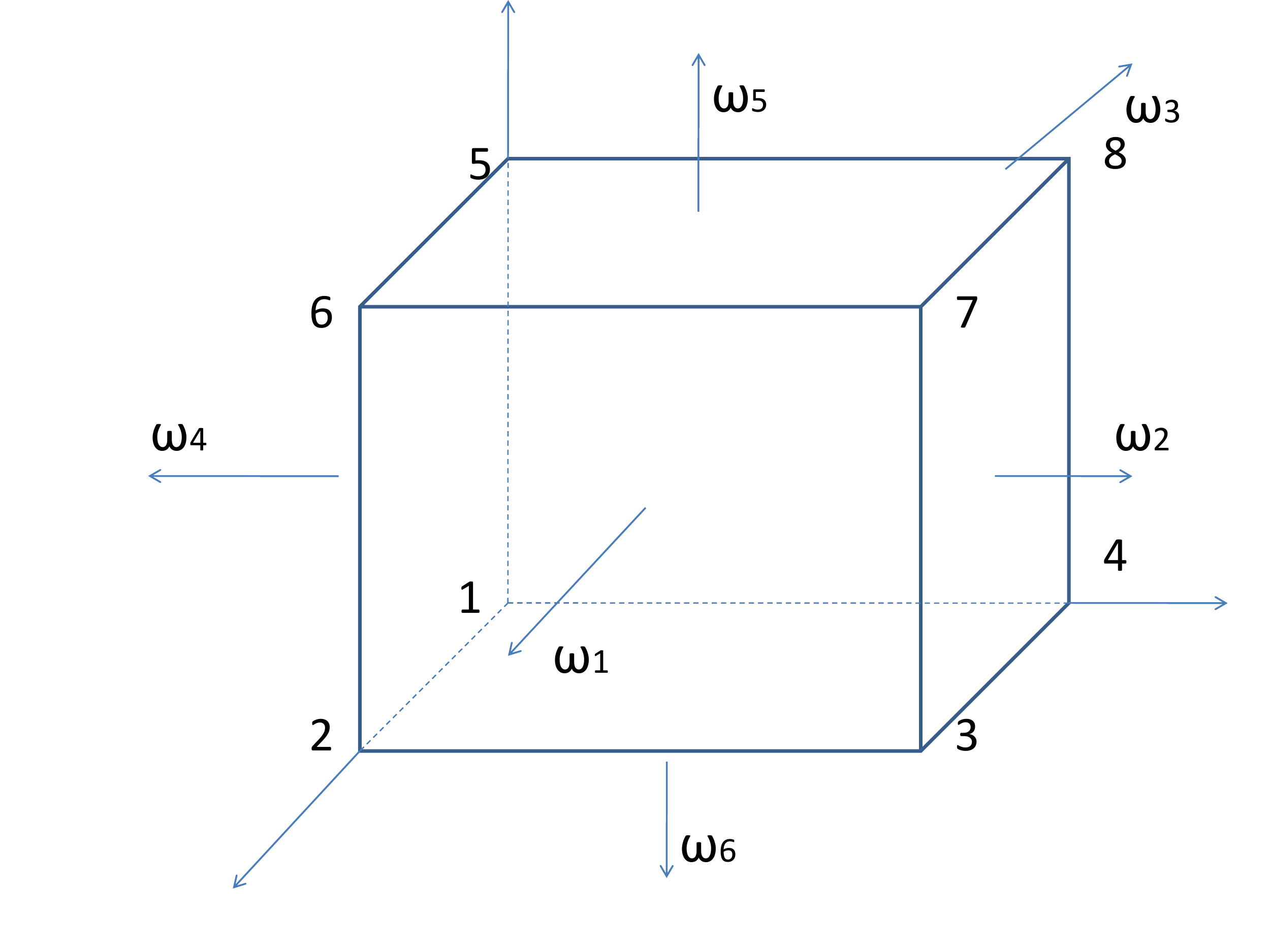}}
       \end{figure}
\centerline{\small{\bf Fig. 4.}}

\section{Lossless versus loss compressions}
Lossless compression
Pure plane based representation of 3D meshes provides lossless compression, since we can exactly restore the initial mesh. As lossless compression its compression ratio depends from polyhedron's shape, and for some polyhedrons it can provide big value of compression.

\noindent In general case the required storage can be calculated for both types of coding as following:
Let P be a convex polyhedron for which the number of faces are $F$ and vertices are $V$.
For coding the polyhedron $\mathbf{P}$ by means of plane based approach will be required $4\cdot3\cdot F = 12F$ bytes.
For coding the polyhedron $\mathbf{P}$ by triangular mesh will be required $4\cdot 3 \cdot V + 4\cdot3\cdot T = (12V + 12T)$ bytes, where $T$ is number of triangles on the faces.  At worst case, when there is only one triangle on each face we would require to store $(12V + 12
F)$ bytes.

\noindent Loss compression
The plane base coding approach additionally allows lossy compression of 3D meshes. The plane based coding is not very sensitive to the removal of some planes from the system of polyhedrons representation, while in contrast to that triangular meshes are quite sensitive to the vertices removal. If a single plane is removed from the plane base representation of a polyhedron it does not change other planes since there is no need of separate topological information.
We suggest the following two algorithms for lossy compression:
	A plane can be removed from the system of representation of a polyhedron if its corresponding face has much smaller surface than other faces. Thus for a given triangle if its area is smaller than certain value $\delta >0$  we can remove its corresponding plane from the representation.
	For two neighbor triangles we can replace their corresponding planes by a single plane if the angle between those planes is quite small (are nearly parallel).  Thus for two given neighbor triangles if the angle between their normals is smaller than certain number $\tau >0$, their corresponding planes can be combined into one plane.	

\section{Representation of a non-convex polyhedron}
We suggest two algorithms for a non-convex polyhedron representation.

1. The first algorithm is the following: we divide a non-convex polyhedron into convex polyhedrons and each of them is represented by its collection of planes.

2. The second algorithm is the following: we divide the surface of a non-convex polyhedron into pseudo-convex and pseudo-concave parts and each of them is represented by its collection of planes.
We denote vertices of an initial triangle as $P_1, P_2, P_3$  and by $e_{ P_1P_2P_3}$ the oriented plane containing that triangle with outside directed normal. We numerate $P_1, P_2, P_3$ in such a way that in the plane $e_{ P_1P_2P_3}$ the insider space of the triangle would be on the left side of the vectors $\overrightarrow{P_1 P_2}$,
$\overrightarrow{P_2 P_3}$,   $\overrightarrow{P_3P_1}$ when we look from the positive hemispace bounded by $e_{ P_1P_2P_3}$.

\noindent Let $P_1 P_2 P_3$  and $Q_1 Q_2 Q_3$  are two triangles. We call them positively oriented to each other if the triangle $P_1 P_2 P_3$ belongs to the negative hemispace bounded by $e_{ Q_1Q_2Q_3}$ and the triangle $Q_1 Q_2 Q_3$ belongs to the negative hemispace bounded by $e_{ P_1P_2P_3}$   (see Fig. 5).

\begin{figure}[here]
\center{\includegraphics[width=9.8cm, height=7.5cm]{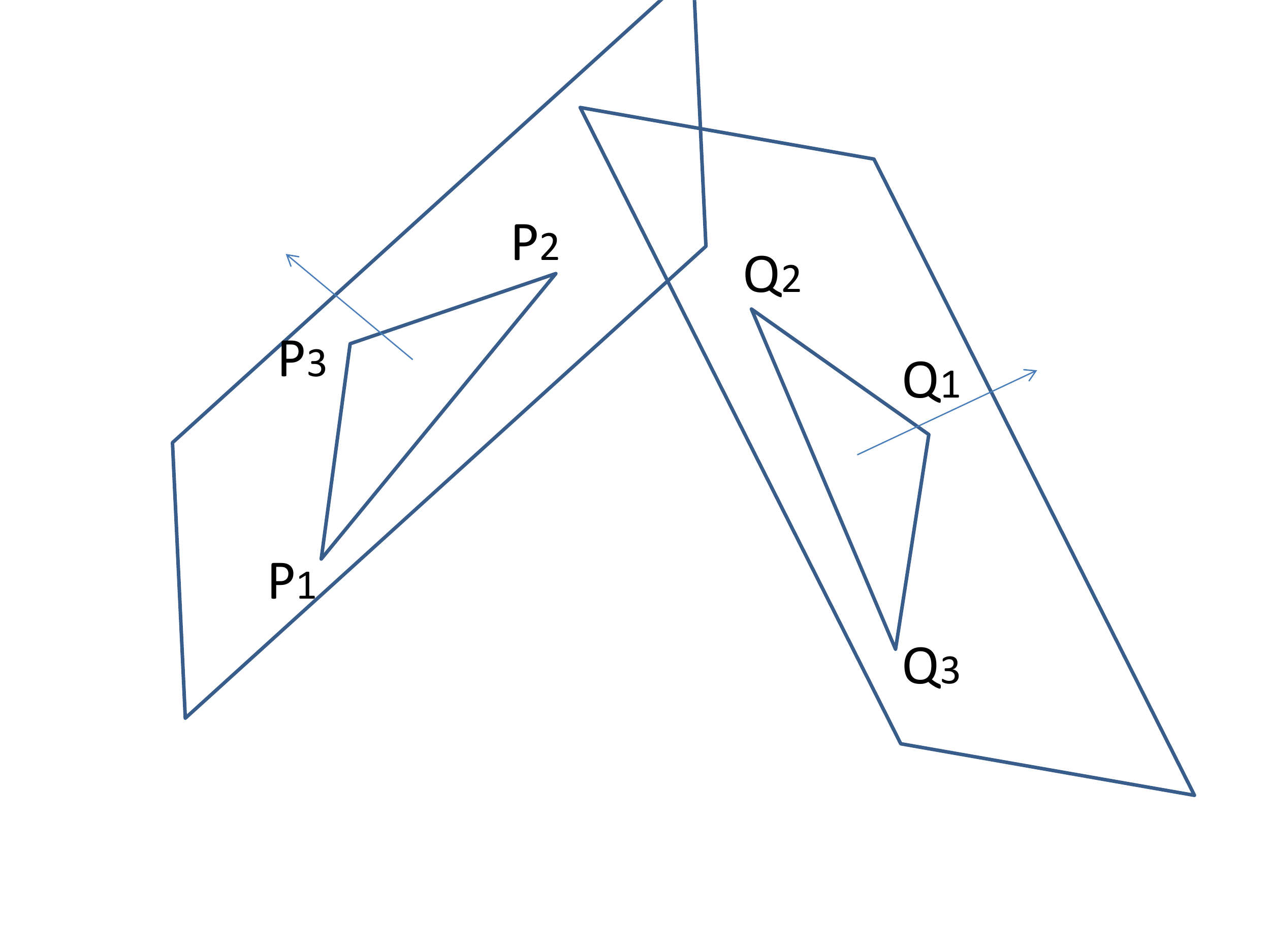}}
       \end{figure}
\centerline{\small{\bf Fig. 5.}}

\noindent  Let $P_1 P_2 P_3$  and $Q_1 Q_2 Q_3$ are two triangles. We call them negatively oriented to each other if the triangle $P_1 P_2 P_3$ belongs to the positive hemi space bounded by $e_{ Q_1Q_2Q_3}$ and the triangle $Q_1 Q_2 Q_3$ belongs to the  positive hemi space bounded by $e_{ P_1P_2P_3}$   (see Fig. 6).

\begin{figure}[here]
\center{\includegraphics[width=9.8cm, height=7.5cm]{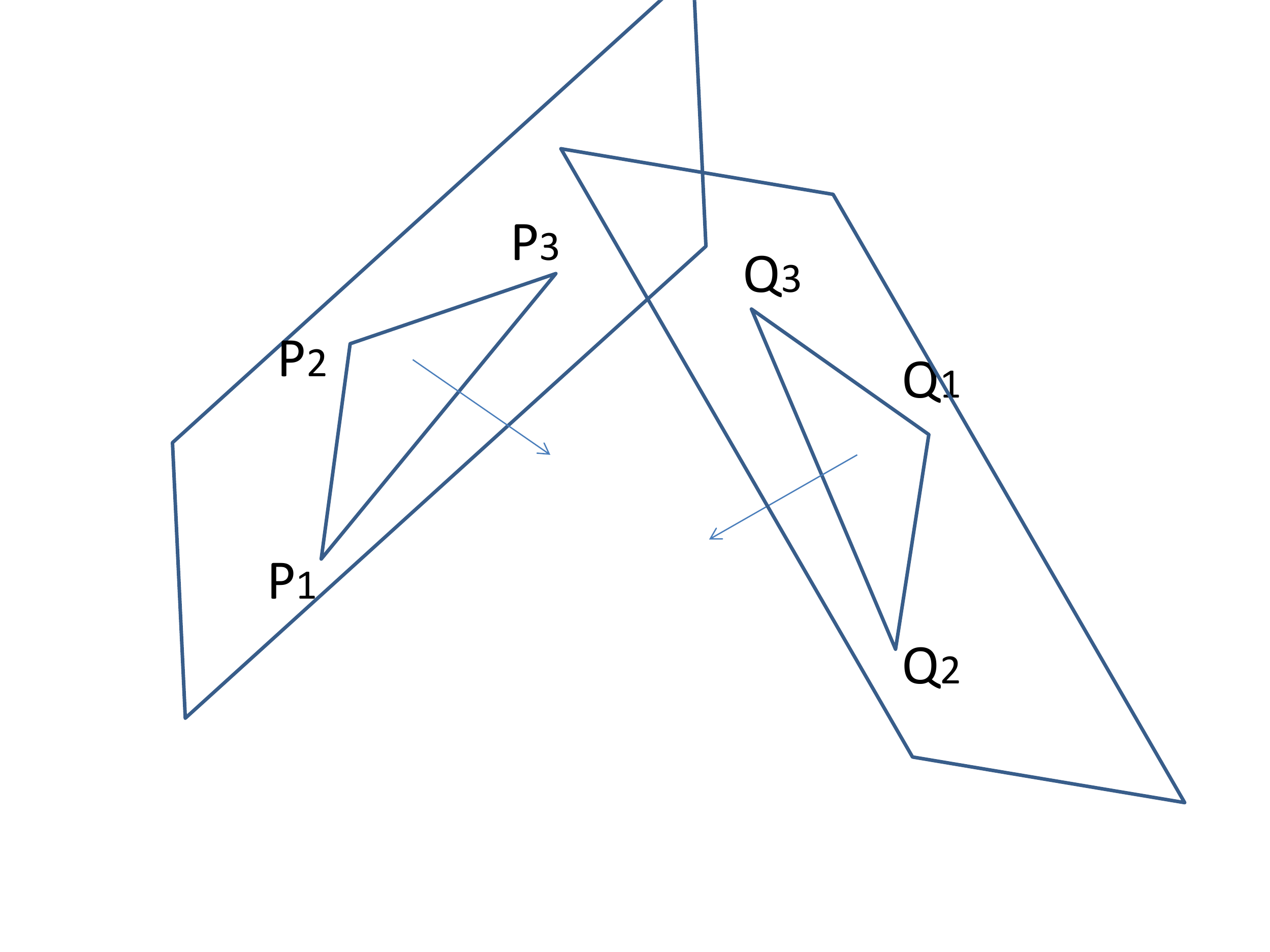}}
       \end{figure}
\centerline{\small{\bf Fig. 6.}}

\emph{ \textbf{Definition 1.} We say that a collection of triangles is pseudo-convex if each two of them are positively oriented to each other .}

\emph{ \textbf{Definition 2.} We say that a collection of triangles is pseudo-concave if each two of them are negatively oriented to each other .}

\noindent We suggest the following algorithm for division of the given non-convex triangular mesh into pseudo-convex and pseudo-concave parts.

\noindent We choose a triangle, which with its neighbor can organize a pseudo-convex collection. Then we verify one by one their neighbor triangles are they positively oriented with all other triangles of the pseudo-convex collection or not, if yes then we add that triangle into the collection and take next neighbor triangle. We end this process if all the neighbors that are not included yet into the pseudo-convex collection are negatively oriented. Thus we get one of the pseudo-convex parts and can remove it from the originally given non-convex mesh.
Then by using same procedure we try to organize another pseudo-convex collection, and after "filling" it with all positively oriented triangles, remove it from the non-convex mesh. We continue this until it becomes impossible to organize any pseudo-convex collection.

\noindent	When all pseudo-convex parts are removed from the original non-convex mesh we start to organize pseudo-concave collections. We choose a triangle, which with its neighbor can organize a pseudo-concave collection. Then we verify one by one their neighbor triangles are they negatively oriented with all other triangles of the pseudo-concave collection or not, if yes then we add that triangle into the collection and take next neighbor triangle. We end this process if all the neighbors that are not included yet into the pseudo-concave collection are positively oriented. Thus we get one of the pseudo-concave parts and can remove it from the originally given non-convex mesh.
Then by using same procedure we try to organize another pseudo-concave collection, and after "filling" it with all negatively oriented triangles, remove it from the non-convex mesh. We continue this until it becomes impossible to organize any pseudo-concave collection.

\noindent Using above presented algorithms we divide the given non-convex mesh into pseudo-convex and pseudo-concave parts.
If occasionally some triangles remain after the above mention procedures of division into pseudo- convex and pseudo-concave parts they can be separately or inside some groups be assumed as pseudo-convex or pseudo-concave parts.

\noindent	When the given non-convex mesh is fully divided into pseudo-convex and pseudo-concave parts we start the process of their polygonization, under which is assumed the unification of all neighbor triangles into polygons if their corresponding oriented planes coincide. After polygonization of   which we denote by $ \{\omega_i,h_i\}_k$, also to each collection we need to include additional planes which will create (cut) the boundaries of the pseudo-convex and pseudo-concave parts since their shapes are not always closed surfaces. Because each of the pseudo-convex and pseudo-concave parts represents the part of the surface of originally given non-convex mesh, after their combination the whole surface can be reconstructed uniquely.
As a result of this procedures the original non-convex polyhedral mesh will be coded by the groups of the planes $ \{\omega_i,h_i\}_k$, where each of the group will represent the certain piece of the whole surface.

\noindent Here we compare the storage requirements of proposed plane based and conventional coding methods.

\begin{figure}[here]
\center{\includegraphics[width=8.8cm, height=6.5cm]{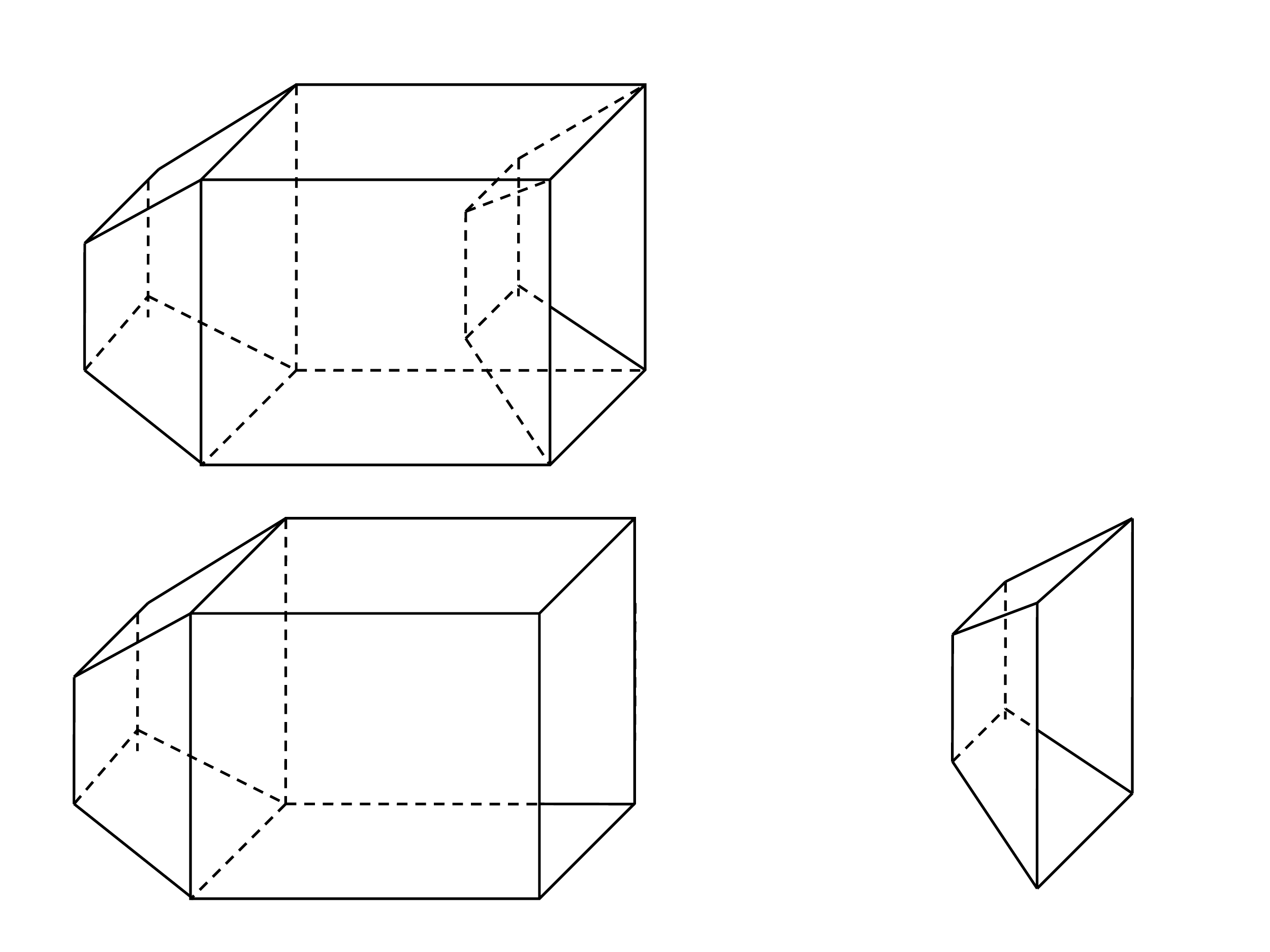}}
       \end{figure}
\centerline{\small{\bf Fig. 7.}}

\noindent Let $\mathbf{P}$ be a non-convex polyhedron shown in Fig. 7.  For its coding by means of plane based approach we divide its surface into pseudo-convex and pseudo-concave parts.
As shown in figure, the pseudo-convex part contains 9 polygons, and additionally we need to include 4 planes for limitation of its bounding borders. So for its storage by means of the planes based approach will be required $3\cdot4\cdot 13 = 156$ bytes.
As shown in figure, the pseudo-concave part contains 5 polygons, and additionally we need to include 4 planes for limitation of its bounding borders. So for its storage by means of the planes based approach will be required $3\cdot 4\cdot 9 = 108$ bytes.
Thus for coding the non-convex whole polyhedron P by means of the plane based approach will be required to store $156 + 108 =264$ bytes.
If the polyhedron $\mathbf{P}$ would be stored by means of conventional quadrangular mesh based approach then we would need to store the following. For storing the data of the 16 vertices would be required $3\cdot4\cdot16 = 192$ bytes. Additionally should be stored the connectivity information among the vertices, which requires the storage of 14 quadrangles such as $4\cdot3\cdot4\cdot14 = 672$ bytes.
Thus for coding the non-convex whole polyhedron P by means of the quadrangular mesh will be required $192 + 672 = 864$ bytes, which is more than three times more than for the plane based storage is required.

\section{Conversion formulas}
In this section we present mathematical formulas for conversion of a triangular 3D mesh into planes based representation.
Let $P_1 P_2 P_3$   be a triangle of a triangular polyhedral mesh, and we should convert it into the plane base representation $(\omega, h) = ( \nu,\varphi, h)$ with outside directed normals. It is easy to mention that $\omega$ is the normalized vector product of the vectors $\overrightarrow{P_1 P_2}$,
$\overrightarrow{P_2 P_3}$ :
$$\omega   =   \frac{\overrightarrow{P_1 P_2}\times\overrightarrow{P_2 P_3}}{|\overrightarrow{P_1 P_2}\times\overrightarrow{P_2 P_3}|}$$
where $|\cdot |$ denotes the length of the vector.
When vertices are given by Euclidean coordinates $P_i = (x_i, y_i, z_i)$, $ i = 1, 2, 3$ by means of  well known  formulas they can be found by the corresponding coordinates $\omega = (\omega_x, \omega_y, \omega_z)$.
It easy to see that $h$ is the scalar product of the vectors $\omega$ and $\overrightarrow{OP_1}$, where $O$ is the origin.
$$h =\omega\cdot \overrightarrow{OP_1}= \omega_x x_1 + \omega_y y_1 + \omega_z z_1.$$

\noindent The spherical coordinates $\nu$, $\varphi$ of the $\omega$ can be calculated by the following formulas:

$$\cos\nu  = \omega_z$$

$$\cos\varphi  = \frac{\omega_x}{\sqrt{1-\omega_z^2}}$$


\begin{thebibliography}{}

\bibitem{S} W. J. Schroeder, "Decimation of Triangle Meshes", Computer Graphics Proceedings, Annual Conference Series, pp. 65-70, ACM SIGGRAPH, 1992.

\bibitem{Tu} Turk, "Re-tiling Polygon Surfaces", Computer Graphics Proceedings, Annual Conference Series, pp.55-64, ACM SIGGRAPH, July, 1992.

\bibitem{Ho} Hoppe et al.,  "Mesh Optimization", Computer Graphics Proceedings, Annual Conference Series, pp.19-26, ACM SIGGRAPH, august, 1992.

\bibitem{De} Deering,  "Geometry Compression", Computer Graphics Proceedings, Annual Conference Series, pp.13-20, ACM SIGGRAPH, august, 1995.

\bibitem{Ta} Taubin, "Geometric Compression Through Topological Surgery", Tech. Rep. Rc-20340, IBM Watson Research Center, January, 1996.

\bibitem{H} Hoppe, "Progressive Meshes", Computer Graphics Proceedings, Annual Conference Series, pp.99-108, ACM SIGGRAPH, august 1996.


\end{thebibliography}
\end{document}